\newif\ifitcs
\itcsfalse

\ifitcs
\PassOptionsToPackage{russian,english}{babel}
\PassOptionsToPackage{capitalize,nameinlink}{cleveref}
\documentclass[a4paper,USenglish,numberwithinsect,cleveref,thm-restate]{lipics-v2021}
\else
\documentclass[11pt]{article}
\usepackage[T1,T2A]{fontenc}
\usepackage[utf8]{inputenc}
\usepackage{fullpage}
\usepackage[margin=1in]{geometry}
\usepackage[shortlabels]{enumitem}
\usepackage[russian,english]{babel}
\fi

\usepackage{bm}
\usepackage{amsfonts}
\usepackage{amssymb}
\usepackage{amsmath}
\usepackage{amsthm}
\usepackage{thmtools, thm-restate}
\usepackage{xcolor,color}
\usepackage{graphicx,caption}
\usepackage{qtree}
\usepackage{tree-dvips}
\usepackage{float}

\ifitcs
\else
\usepackage[pagebackref,colorlinks=true,linkcolor=blue,urlcolor=blue,citecolor=blue,pdfstartview=FitH]{hyperref}
\usepackage[capitalize,nameinlink]{cleveref}
\fi

\usepackage[ruled,vlined]{algorithm2e}
\usepackage{epigraph}

\usepackage[colorinlistoftodos]{todonotes}
\ifitcs
\newcommand{\dmnote}[1]{}
\newcommand{\phnote}[1]{}
\newcommand{\arnote}[1]{}
\newcommand{\sriprahladhuvacha}[1]{}
\else
\newcommand{\dmnote}[1]{\todo[color=red!100!green!33, size=\footnotesize]{Dan: #1}}
\newcommand{\phnote}[1]{\todo[color=red!100!green!33, size=\footnotesize]{ph: #1}}
\newcommand{\arnote}[1]{\todo[color=red!20!blue!15, size=\footnotesize]{Alon: #1}}
\newcommand{\sriprahladhuvacha}[1]{\todo[color=red!100!green!33,inline,size=\small]{ph: #1}}
\fi

\makeatletter
\renewcommand{\@algocf@capt@plain}{above}
\makeatother

\ifitcs
\else 
\newtheorem{theorem}{Theorem}[section]
\newtheorem*{theorem*}{Theorem}
\newtheorem{lemma}[theorem]{Lemma}
\newtheorem*{lemma*}{Lemma}
\newtheorem{corollary}[theorem]{Corollary}
\newtheorem{proposition}[theorem]{Proposition}

\theoremstyle{definition}

\newtheorem{definition}[theorem]{Definition}

\fi 
\newtheorem*{notation}{Notation}
\newtheorem{openproblem}[theorem]{Open Problem}

\newcommand{\B}{\{0,1\}}
\newcommand{\bits}{\B}

\newcommand{\NP}{\textsc{NP}}
\newcommand{\coNP}{\textsc{co-NP}}

\newcommand{\TFNP}{\textsc{TFNP}}
\newcommand{\TFUP}{\textsc{TFUP}}
\newcommand{\PPAD}{\textsc{PPAD}}
\newcommand{\PLS}{\textsc{PLS}}
\newcommand{\PPP}{\textsc{PPP}}

\newcommand{\CLS}{\textsc{CLS}}
\newcommand{\UEOPL}{\textsc{UEOPL}}

\newcommand{\PSPACE}{\textsc{PSPACE}}
\newcommand{\FACTOR}{\textsc{Factor}}
\newcommand{\ALLFACTORS}{\textsc{AllFactors}}
\newcommand{\email}[1]{Email: \href{mailto:#1}{\texttt{#1}}}

\ifitcs

\title{Downward Self-Reducibility in TFNP}

\author{Prahladh Harsha}{Tata Institute of Fundamental Research, Mumbai, India \and \url{https://www.tifr.res.in/~prahladh/}}{prahladh@tifr.res.in}{https://orcid.org/0000-0002-2739-5642}{Research supported by the Department of Atomic
Energy, Government of India, under project 12-R\&D-TFR-5.01-0500.\flag{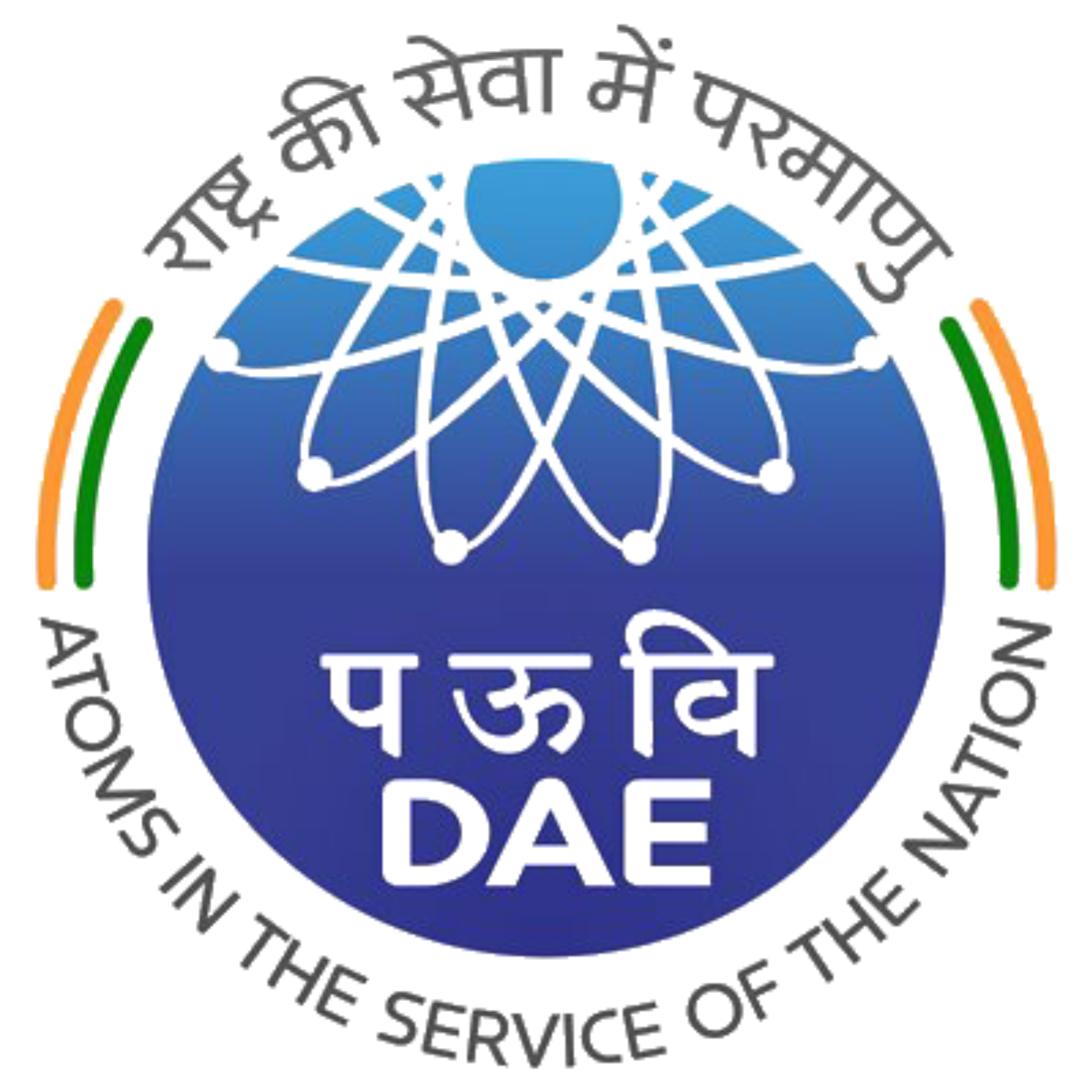}}

\author{Daniel Mitropolsky}{Columbia University, USA \and \url{https://dmitropolsky.github.io/}}{mitropolsky@cs.columbia.edu}{}{Supported in part by the Columbia-IBM center for Blockchain and Data Transparency, by LexisNexis Risk Solutions, and by European Research Council (ERC) under the European Union’s Horizon 2020 research and innovation programme (Grant agreement No. 101019547). \flag{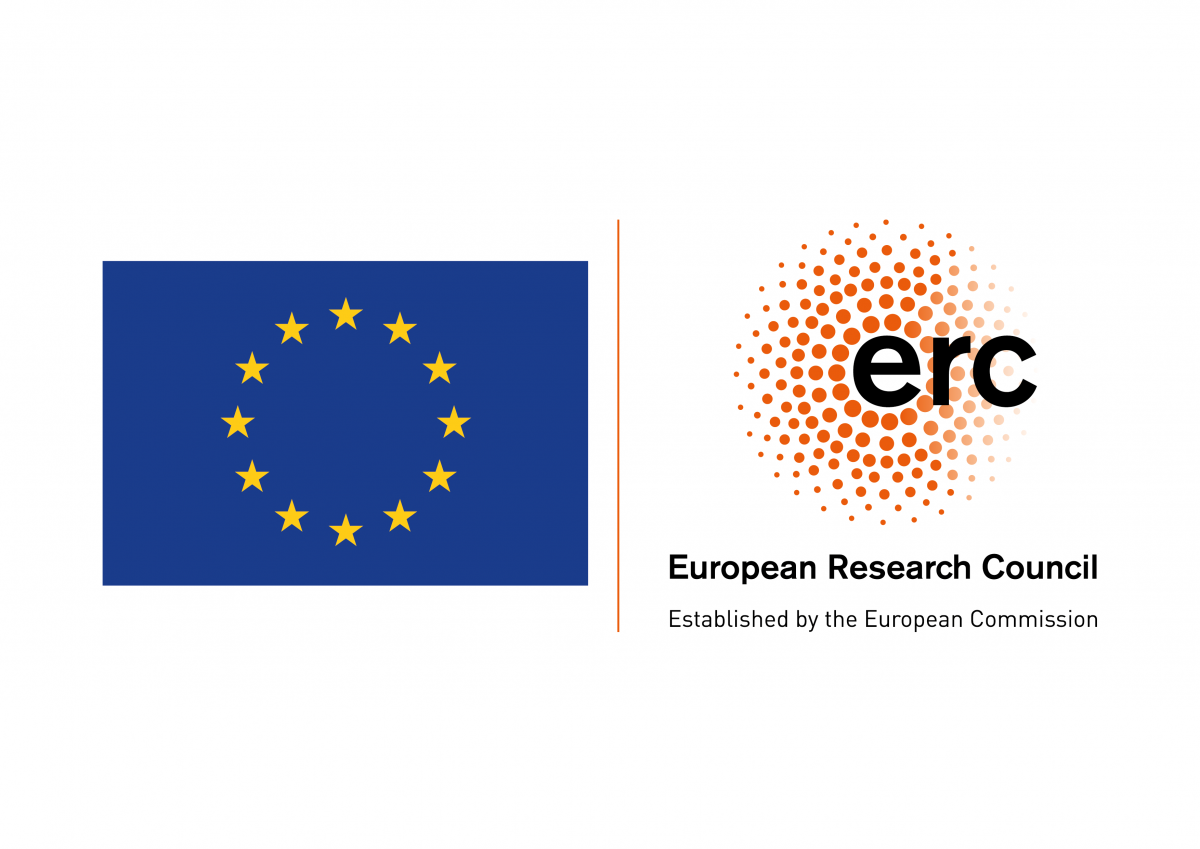}}

\author{Alon Rosen}{Bocconi University, Italy \and Reichman University, Israel \and \url{https://www.alonrosen.net/}}{alon.rosen@unibocconi.it}{https://orcid.org/0000-0002-3021-7150}{Supported by the European Research Council (ERC) under the European Union’s Horizon 2020 research and innovation programme (Grant agreement No. 101019547) and Cariplo CRYPTONOMEX grant. \flag{ERC_EU.png}}
    
\authorrunning{P. Harsha, D. Mitropolsky, and A. Rosen} 

\Copyright{Prahladh Harsha, Daniel Mitropolsky, and Alon Rosen} 



\category{} 

\relatedversiondetails{arXiv Version}{https://arxiv.org/abs/2209.10509} 

\acknowledgements{This work was initiated when the first and second authors were visiting the third author at Bocconi University and we are thankful to Boccconi University for their hospitality. We thank Pavel Hub{\'{a}}cek, Eylon Yogev, and Omer Paneth for their comments on an earlier draft of this paper. We also thank the anonymous referees for several useful comments and informing us of \textsc{End-Of-Potential-Line} and \textsc{UniqueEOPL}, the complete problems for the classes \textsc{CLS} and \textsc{UEOPL} respectively, which simplified certain parts of our proof.}

\EventEditors{Yael Tauman Kalai}
\EventNoEds{1}
\EventLongTitle{14th Innovations in Theoretical Computer Science Conference (ITCS 2023)}
\EventShortTitle{ITCS 2023}
\EventAcronym{ITCS}
\EventYear{2023}
\EventDate{January 10--13, 2023}
\EventLocation{MIT, Cambridge, Massachusetts, USA}
\EventLogo{}
\SeriesVolume{251}
\ArticleNo{89}

\else 

\title{Downward Self-Reducibility in TFNP\thanks{This work was initiated when the first and second authors were visiting the third author at Bocconi University and we are thankful to Boccconi University for their hospitality. We thank Pavel Hub{\'{a}}cek, Eylon Yogev, and Omer Paneth for their comments on an earlier draft of this paper. We also thank the anonymous referees for several useful comments and informing us of \textsc{End-Of-Potential-Line} and \textsc{UniqueEOPL}, the complete problems for the classes \textsc{CLS} and \textsc{UEOPL} respectively, which simplified certain parts of our proof.}}

\author{Prahladh Harsha\thanks{Tata Institute of Fundamental Research, India. \email{prahladh@tifr.res.in}. Research supported by the Department of Atomic
Energy, Government of India, under project 12-R\&D-TFR-5.01-0500.}
    \and 
        Daniel Mitropolsky\thanks{Columbia University, USA. \email{mitropolsky@cs.columbia.edu}. Supported in part by the Columbia-IBM center for Blockchain and Data Transparency, by LexisNexis Risk Solutions, and by European Research Council (ERC) under the European Union’s Horizon 2020 research and innovation programme (Grant agreement No. 101019547).}
    \and 
    Alon Rosen\thanks{Bocconi University, Italy and Reichman University, Israel. \email{alon.rosen@unibocconi.it}. Supported by the European Research Council (ERC) under the European Union’s Horizon 2020 research and innovation programme (Grant agreement No. 101019547) and Cariplo CRYPTONOMEX grant.}
    }
\fi 

\begin{document}

\maketitle

\begin{abstract}
A problem is \emph{downward self-reducible} if it can be solved efficiently given an oracle that returns solutions for strictly smaller instances. In the decisional landscape, downward self-reducibility is well studied and it is known that all downward self-reducible problems are in \textsc{PSPACE}.  In this paper, we initiate the study of downward self-reducible search problems which are guaranteed to have a solution --- that is, the downward self-reducible problems in \textsc{TFNP}. We show that most natural $\PLS$-complete problems are downward self-reducible and any downward self-reducible problem in \textsc{TFNP} is contained in \textsc{PLS}.
Furthermore, if the downward self-reducible problem is in \textsc{TFUP} (i.e. it has a unique solution), then it is actually contained in \textsc{UEOPL}, a subclass of \textsc{CLS}. This implies that if integer factoring is \emph{downward self-reducible} then it is in fact in \textsc{UEOPL}, suggesting that no efficient factoring algorithm exists using the factorization of smaller numbers.
\end{abstract}

\section{Introduction}

\epigraph{{Perhaps the most surprising thing about Self-Reducibility is its longevity.}}{Eric Allender}

Self-reducibility (sometimes also referred to as auto-reducibility) asks the question if a problem instance is easy to solve if one can solve \emph{other} instances of the problem easily. More precisely, a computational problem is said to be self-reducible if there exists an efficient oracle machine that on input an instance of the problem solves it by making queries to an oracle that solves the same problem, with the significant restriction that it cannout query the oracle on the given instance. First introduced in the context of computability by Trakhtenbrot \cite{Trakhtenbrot1970} more than five decades ago, self-reducibility has proved to be immensely useful in computational complexity theory, especially in the study of decision problems. Various notions of self-reducibility have been studied (for a detailed survey on self-reducibility, the reader is referred to Balc{\'{a}}zar's systematic study \cite{Balcazar1990}, Selke's comprehensive survey \cite{Selke2006} and Allender's survey \cite{Allender2010}).

Downward self-reducibility, the key notion studied in this paper, refers to the kind of self-reducibility wherein the oracle machine is allowed to query the oracle only on instances strictly \emph{smaller} than the given instance. Satisfiability, the prototype $\NP$-complete problem as well as several other related $\NP$-complete problems can be easily shown to be downward self-reducible. This simple observation led to plethora of search-to-decision reductions for NP-complete problems. All $\NP$-complete problems are self-reducible and many natural $\NP$-complete problems are also downward self-reducible. Furthermore, it is known that every downward self-reducible problem is contained in $\PSPACE$ (folklore). Another notion of self-reducibility that is extremely well-studied is random self-reducibility; wherein the oracle machine (which is randomized in this case) is only allowed to query the oracle on \emph{random} instances of the problem, under a natural distribution on the input instance. There are problems in $\PSPACE$ which are known to be both downward self-reducible as well as random self-reducible, in particular permanent, one of the \textsc{\#P}-complete problems.  The downward and random self-reducibility of the permanent and other $\PSPACE$-complete problems led to efficient interactive proofs for $\PSPACE$.

Almost all the study of downward self-reducibility, to date, has been focused in the decisional landscape. In this work, we initiate the study of downward self-reducibility for search problems which are guaranteed to have a solution, that is problems in \TFNP. This is motivated by the oft asked, but yet unresolved, question regarding integer factoring: is factoring an integer any easier if we know how to factor all \emph{smaller} integers? Or equivalently, is factoring downward self-reducible?

\subsection{Our results}

Our focus will be on the class of search problems for which there is guaranteed to be a solution, more concisely referred to as $\TFNP$. Within $\TFNP$, there are several sub-classes of search problems such as $\PLS$, $\textsc{PPAD}$, $\textsc{PPP}$, $\textsc{PPA}$, $\CLS$, for which the totality -- the fact that every instance has an efficiently verifiable solution -- arises from some underlying combinatorial principle. These problems are typically defined in terms of natural complete problems. For instance, the class $\PLS$, which is a subclass of $\TFNP$, containing problems where totality is guaranteed by a ``local-search'' principle, is characterized by its complete problems \textsc{ITER} and \textsc{Sink-of-DAG}. Our first result shows that several of these natural complete problems for $\PLS$ are in fact downward self-reducible.

\begin{theorem} \label{thm:dsr}
\textsc{ITER}, \textsc{Sink-of-DAG} are downward self-reducible.
\end{theorem}

This result is similar to the corresponding situation for $\NP$ and $\PSPACE$. Most of the naturally occuring $\NP$-complete and $\PSPACE$-complete problems are downward self-reducible. It is open if \emph{all} $\NP$-complete or $\PSPACE$-complete problems are downward self-reducible. The situation is similar for $\PLS$: it is open if every $\PLS$-complete problem is downward self-reducible.

Given the above result, one may ask if there are problems outside $\PLS$ in $\TFNP$ which are downward self-reducible. Our next result rules this out.

\begin{restatable}{theorem}{dsrinPLS}\label{thm:dsr-in-PLS}
  Every downward self-reducible problem in $\TFNP$ is in $\PLS$.
\end{restatable}

That is, just as $\PSPACE$ is the ceiling for downward self-reducible problems in the decisional world, $\PLS$ contains all downward self-reducible problems in $\TFNP$. An immediate corollary of these two results is that a downward self-reducible problem in $\TFNP$ is hard iff a problem in $\PLS$ is hard.

We then ask if we can further characterize the downward self-reducible problems in $\TFNP$: what about search problems which are guaranteed to have an \emph{unique} solution, more concisely denoted as $\TFUP$. In this case, we show that the above containment of d.s.r problems in $\TFNP$ in the class $\PLS$ can be strengthened if the search problem is guaranteed to have an unique solution: the problem then collapses to \textsc{UEOPL}, a subclass of $\CLS$ (itself a sub-class of $\PLS$), which is one of the lowest subclasses known in the $\TFNP$ hierarchy.

\begin{restatable}{theorem}{tfup}\label{thm:tfup-cls}
Every downward self-reducible problem in $\TFUP$ is in $\UEOPL$.
\end{restatable}

  We now return to the question on downward self-reducibility of integer factoring. Let $\FACTOR$ denote the problem of finding a non-trivial factor of the given number if the number is composite or outputting ``prime'' if the number is prime and $\ALLFACTORS$ be the problem of listing \emph{all} the non-trivial factors of a given number if it is composite or outputting ``prime'' if the number is prime. Clearly, $\ALLFACTORS$ is in $\TFUP$. It is not hard to check that $\ALLFACTORS$ is downward self-reducible iff $\FACTOR$ downward self-reducible. Combining this with the above result, we get that factoring collapses to $\UEOPL$ if factoring is downward-self reducible. This result can be viewed as evidence against the downward self-reducibility of factoring.

\begin{restatable}{corollary}{factor}\label{cor:factor}
If \textsc{Factor} or \textsc{AllFactors} is downward self-reducible, then \textsc{Factor} and \textsc{AllFactors} $\in \UEOPL$. 
\end{restatable}

\subsection{Related work}
The class $\TFNP$ was introduced by Meggido and Papadmitriou \cite{MegiddoP1991}. In 1994, Papadimitriou initated the study of syntactic subclasses of $\TFNP$, i.e., search problems where totality is guaranteed by combinatorial principles, including $\PPAD$ \cite{Papadimitriou1994}.

In the original 1990 paper, Meggido and Papadimitriou proved that if $\TFNP$ contains an $\NP$-hard problem, then the polynomial hierarchy collapses to first level, i.e. $\NP = \coNP$. Since this is generally believed not to be true, $\TFNP$ has emerged as an important potential source of ``intermediate"  problems -- that is, problems that are not in $\textsc{P}$ but not $\NP$-complete. Since our work gives a new characterization of $\PLS$-hardness, it is related to work that shows $\PLS$-hardness from cryptographic assumptions, such as in \cite{BitanskyPR2015,GargPS2016,HubacekY2020}. In addition to hardness from cryptographic assumptions, Bitansky and Gerichter recently showed unconditional hardness of $\PLS$ relative to random oracles, as well as hardness of $\PLS$ from incrementally verifiable computation (though an instantiation of this is given through cryptographic assumptions) \cite{BitanskyG2020}. 

The connection between factoring and $\TFNP$ has been studied by Buresh-Oppenheim \cite{Buresh-Oppenheim2006} and later by Jerabek \cite{Jerabek2016}. Their work shows that factoring is reducible to $\textsc{PPA}~\cap~\PPP$. It is known that $\PPAD \subseteq \textsc{PPA}~\cap~\PPP$, but whether factoring is in $\PPAD$ is an open problem. This paper makes progress in the study of factoring and $\TFNP$ by showing that if a downward self-reduction for factoring exists, factoring is not only in $\PPAD$ but also in $\CLS$. $\CLS$ was introduced by Daskalakis and Papadimitriou in 2011 \cite{DaskalakisP2011}, and in a recent breakthrough result, shown to be exactly the intersection $\PPAD~\cap~\PLS$ \cite{FearnleyGHS2021}. Recently, it was also shown that $\CLS$ is hard relative to a random oracle but with the additional assumption of \#SAT hardness \cite{ChoudhuriHKPRR2019-fs}. More recently, in concurrent and independent work to ours, it was shown that $\UEOPL$ (and hence $\CLS$) hardness can be based on any downward self-reducible problem family that additionally satisfied a randomized batching property, and for which the Fiat-Shamir heuristic protocol is unambiguously sound \cite{BitanskyCHKLPR22}. This complements our results in interesting ways; we base hardness of $\UEOPL$ on any problem in unique $\TFNP$ that is d.s.r., while theirs is based on {\em any} problem that is d.s.r. (not necessarily total) but requires an additional batching property and an assumption about the Fiat-Shamir heuristic for interactive protocols for that problem. The reductions in \cref{thm:dsr,thm:dsr-in-PLS} are reminiscent of the reductions in \cite{ChoudhuriHKPRR2019-ppad} and \cite{EphraimFKP2020} in their reliance on downward self-reducibility. The major differences between these works and our work are 
(i) these results work with specific d.s.r. problems (viz., sumcheck of low-degree polynomials and iterated squaring) while we work with any d.s.r. problem in $\TFNP$; and
(ii) these results rely on problems ``outside'' of $\NP$ while we rely on problems in $\TFNP$.
This allows us to do away with non-interactive proofs/arguments that the aforementioned works rely on (since a solution of a $\TFNP$ problem is itself a proof) and generalise/abstract out the assumption (i).

\subsection*{Organization} 

The rest of this paper is organized as follows. We first begin with some preliminaries in \cref{sec:prelim} where we recall the definitions of various subclasses of problems in $\TFNP$ and the notion of downward self-reducibility. In \cref{sec:dsr-problems}, we show that some of the classical $\PLS$-complete problems such as \textsc{ITER}, \textsc{Sink-of-Dag} are downward self-reducible, thus proving \cref{thm:dsr}. Next, in \cref{sec:dsrinPLS}, we show that every downward self-reducible problem in $\TFNP$ is contained in $\PLS$, thus proving \cref{thm:dsr-in-PLS}. In \cref{sec:tfup}, we observe that if the search problem is in $\TFUP$, then the reduction in \cref{sec:dsrinPLS} can be strengthened to show that the problem is contained in $\CLS$ (thus, proving \cref{thm:tfup-cls}). This immediately yields \cref{cor:factor} on the consequence of the downward self-reducibility of factoring. We finally conclude with some closing thoughts in \cref{sec:disc}.

\section{Preliminaries}\label{sec:prelim}

\begin{notation}
Throughout this paper $||$ denotes concatenation of strings. We impose an ordering on strings in $\bits^n$ lexicographically, e.g. for $x=1011$ and $y=1001$, $x > y$. Sometimes we will write $[T]$ to represent the subset of strings $\bits^{\lceil\log(T)\rceil}$ that are at most $T$ as binary numbers.
\end{notation}

\subsection{Search Problems}
We begin by recalling the definitions of search problems, $\TFNP$, and its subclasses.

\begin{definition}[Search problem] A search problem is a relation $R \subset \{0,1\}^* \times \{0,1\}^*$. We will variously write $xRy$, $(x,y) \in R$ and $R(x,y) = 1$ to mean that $(x,y)$ is in the relation $R$, or equivalently that $y$ is a solution to $x$ in the search problem $R$. A Turing Machine or algorithm $M$ solves $R$ if for any input $x$, $M(x)$ halts with a string $y$ on the output tape such that $(x,y) \in R$.
\end{definition}

\begin{definition}[NP search problem] A search problem $R$ is an NP search problem if (1) there exists a polynomial $p(\cdot)$ such that for every $x \in \{0,1\}^*$, the set $Y_x = \{y \in \B^*~:~(x,y)\in R\} \subseteq \B^{p(n)}$ and (2) there exists a polynomial time Turing Machine $M$ such that for all pairs of strings $x,y$, $M(x,y)=1 \iff (x,y) \in R$. The class of NP search problems is denoted FNP.
\end{definition}

\begin{definition}[TFNP] A search problem $R$ is \emph{total} if for every $x \in \B^*$, there exists at least one $y$ such that $(x,y)\in R$. The class of total NP search problems is called $\TFNP$.
\end{definition}

\subsection{Local Search and Continuous Local Search}

The study of total complexity is primarily interested in problems where totality -- the fact that for any input there is a solution (that can be efficiently verified) -- is the result of some combinatorial principle. A fundamental subclass of this kind is $\PLS$, which intuitively captures problems whose totality is guaranteed by a ``local search"-like principle (that is, the principle that if applying some efficient step increases an objective but the objective is bounded, a local optimum must exist).

\begin{definition}[\textsc{Sink-of-DAG}]
Given a successor circuit $S:\B^n\rightarrow\B^n$ such that $S(0^n) \neq 0^n$, and a valuation circuit $V:\B^n\rightarrow [2^m-1]$ find $v \in \B^n$ such that $S(v) \neq v$ but $S(S(v)) = S(v)$ (a sink), or $V(S(v)) \leq V(v)$.
\end{definition}

It is not hard to show that \textsc{Sink-of-DAG}~$\in\TFNP$: since the range of $V$ is finite, successively applying $S$ must result in a sink of $S$, or a local maximum of $V$ (note that exponentially many applications of $S$ may be required to find this solution, so we cannot say that $\PLS$ is solvable in polynomial time). Any problem that is polynomial-time reducible to \textsc{Sink-of-DAG} would also be in $\TFNP$.

\begin{definition}[PLS \cite{JohnsonPY1988}]
The class $\PLS$, for \emph{polynomial local search}, consists of all search problems in $\TFNP$ that are polynomial-time reducible to \textsc{Sink-of-DAG}.
\end{definition}

In this paper, we will be interested in several other complete problems in $\PLS$.

\begin{definition}[\textsc{ITER} \cite{Morioka2005}] Given a successor circuit $S:\B^n\rightarrow\B^n$ such that $S(0^n) > 0^n$, find $v \in \B^n$ such that $S(v) > v$, but $S(S(v)) \leq S(v)$.
\end{definition}

\begin{definition}
\textsc{ITER-with-source}: Given a successor circuit $S:\B^n\rightarrow\B^n$ and a source $s \in \B^n$ such that $S(s) > s$, find $v \in \B^n$ such that $S(v) > v$, but $S(S(v)) \leq S(v)$.
\end{definition}

Analogously to \textsc{ITER} versus \textsc{ITER-with-source}, one can define \textsc{Sink-of-DAG-with-source}, where instead of the guarantee that $S(0^n) \neq 0^n$, we are given an additional input $s \in \B^n$ such that $S(s) \neq s$. The following is folklore:

\begin{theorem} \label{thm:cdsr-problems}
All of \textsc{ITER}, \textsc{ITER-with-source}, \textsc{Sink-of-DAG}, \textsc{Sink-of-DAG-with-source} are $\PLS$-complete.
\end{theorem}
\begin{proof}
\begin{itemize}
\item \textbf{\textsc{ITER}~$\rightarrow$~\textsc{Sink-of-DAG}:} Given an instance $S$ of \textsc{ITER}, let $S'=S$ and $V'=S$; a solution $v$ of $(S',V')$ is trivially a solution of $S$.
\item \textbf{\textsc{Sink-of-DAG}~$\rightarrow$~\textsc{ITER}:} Given an instance $(S,V)$ of \textsc{Sink-of-DAG}, let $S':\B^{2n}\rightarrow \B^{2n}$ be the circuit computing $S'(V(x)||x)) = V(S(x))||S(x)$ and $S'(y||x) = y||x$ if $y\neq V(x)$, and source $s' = V(0^n)||0^n$. A solution $v$ of $S'$ must satisfy that either $S(S(v))=S(v)$, or $V(S(S(v)) \leq S(v)$, so either $v$ or $S(v)$ is a solution to $S$.
\item \textbf{\textsc{X}~$\rightarrow$~\textsc{X-with-source}, where \textsc{X}~$\in$~\{\textsc{ITER}, \textsc{Sink-of-DAG}\}:} to specify the additional argument $s$, let $s=0^n$.
\item \textbf{\textsc{X-with-source}~$\rightarrow$~\textsc{X}, where \textsc{X}~$\in$~\{\textsc{ITER}, \textsc{Sink-of-DAG}\}:} replace the successor circuit $S$ of either problem with $S'$ that computes $S'(0^n) = s$ and $S'(x) = x$ for $x\neq 0^n$.
\end{itemize}
\end{proof}

Another fundamental class in the study of total complexity is $\PPAD$, which similar to $\PLS$ is a sink-(or source)-finding problem, but where the existence of sinks is not guaranteed by a consistently increasing valuation, but rather by the presence of a predecessor circuit. In both $\PLS$ and $\PPAD$, a principle (either the increasing valuation or the predecessor circuit) ensures that loops are not possible along a path.
\begin{definition}[\textsc{End-of-line}]
In \textsc{End-of-Line} (or \textsc{EOL}) problem, we are given a \emph{successor} circuit $S:\bits^n\rightarrow \bits^n$ and a \emph{predecessor} circuit $P:\bits^n\rightarrow \bits^n$ such that $S(0^n) \neq 0^n$ and $P(0^n) = 0^n$ and we need to output either another \emph{source} node ($v$ s.t. $S(v) \neq v$ and $P(v) = v$) or a \emph{sink} node ($v$ s.t. $P(v) \neq v$ but $S(v) = v$).
\end{definition}

\begin{definition}[$\PPAD$ \cite{Papadimitriou1994}]
The class $\PPAD$, for \emph{polynomial parity arguments on directed graphs}, consists of all search problems in $\TFNP$ that are polynomial-time reducible to \textsc{End-of-Line}.
\end{definition}

In this paper we will be interested in the class $\CLS$, which was originally defined with respect to a \emph{continuous} analogue of $\PLS$ and $\PPAD$. However, it has recently been shown to have a much simpler complete problem which is the natural unification of $\PLS$ and $\PPAD$:




\begin{definition}[\textsc{End-of-Potential-Line} \cite{FearnleyGMS2020}] Given a \emph{successor} circuit $S:\bits^n\rightarrow \bits^n$ and a \emph{predecessor} circuit $P:\bits^n\rightarrow \bits^n$ such that $S(0^n) \neq 0^n$ and $P(0^n) = 0^n$, as well as a valuation circuit $V:\B^n\rightarrow [2^m-1]$, find either
\begin{itemize}
\item a ``$\PPAD$-like solution": $v$ s.t. $S(v)\neq v = P(v)$ or $P(v)\neq v = S(v)$
\item a ``$\PLS$-like solution": $v$ s.t. $S(v)\neq v$, $P(S(v))=v$ and $V(S(v)) \leq V(v)$.
\end{itemize}
\end{definition}

\begin{definition}[$\CLS$, defined in \cite{DaskalakisP2011}, completeness of \textsc{End-of-Potential-Line} proved in \cite{GoosH0MPRT2022}]
The class $\CLS$, for \emph{continuous local search}, is the class of total search problems polynomial time reducible to \textsc{End-Of-Potential-Line}.
\end{definition}

A recent result surprised the $\TFNP$ community:
\begin{theorem}[\cite{FearnleyGHS2021}]
$\CLS = \PLS \cap \PPAD$.
\end{theorem}

We think of \textsc{End-of-Line} and \textsc{End-of-Potential-Line} instances as defining an ordered ``line" consisting of the vertices $0^n,S(0^n),S^2(0^n),\ldots$. If there were exactly one ``line" in the instance, the solution to these problems would be unique. We will be interested in a version of \textsc{End-of-Potential-Line} that allows an additional kind of solution, namely one which demonstrates the existence of at least two lines.

\begin{definition}[\textsc{UniqueEOPL} \cite{FearnleyGMS2020}] Given an instance $(S,P,V)$ of \textsc{End-of-Potential-Line} with vertex set $\bits^n$, find a solution to the original problem \emph{or} two points $u~\neq~v\in\bits^n$ such that $x \neq S(x)$, $y \neq S(y)$ and either $V(x)=V(y)$ or $V(x) < V(y) < V(S(x))$.
\end{definition}

Observe that either $V(x)=V(y)$ or $V(x) < V(y) < V(S(x))$ in the new type of solution above implies the existence of at least two ``lines" in the graph represented by $(S,P,V)$. This problem leads to the definition of a natural subclass of $\CLS$: 

\begin{definition}[\textsc{UEOPL} \cite{FearnleyGMS2020}]
The class \textsc{UEOPL} is the class of total search problems polynomial time reducible to \textsc{UniqueEOPL}.
\end{definition}

\begin{proposition}[\cite{FearnleyGMS2020}]
\textsc{UEOPL} is a subclass of $\CLS$.
\end{proposition}

\subsection{Downward-self-reducibility}\label{sec:dsr}
\begin{definition}
A search problem $R$ is \emph{downward self-reducible} (d.s.r) if there is a polynomial time oracle algorithm for $R$ that on input $x \in \B^n$ makes queries to an $R$-oracle of size strictly less than $n = |x|$.
\end{definition}

In this definition of downward self-reducible, the number of \emph{bits} of the input must be smaller in oracle calls. Indeed, both \textsc{TQBF} and \textsc{SAT} are d.s.r.\ in this sense by the usual trick of instantiating the first quantified variable $x_1$, \emph{and simplifying the resulting Boolean expression} to remove $x_1$ entirely. However, there is a related notion of downward self-reducibility which can be defined to certain classes of structured problems or languages, where the \emph{number of variables} of the instance decreases in oracle calls, even if the number of bits of the sub-instances does not decrease. A circuit problem is one where inputs are always circuits $C:\B^n \rightarrow \B^m$ under some fixed representation of circuits, and solutions are inputs to the circuit of the form $s \in \B^n$. 

\begin{definition}
A circuit problem is circuit-d.s.r. if there is a polynomial time oracle algorithm for $R$ that on input $C:\B^n \rightarrow \B^m$ makes queries to an $R$-oracle on circuits with $\nu \leq n$ input bits, and $\mu \leq m$ output bits, and $\nu+\mu < n+m$ (at least one of $\nu$ or $\mu$ is smaller than $n$ or $m$ respectively).
\end{definition}

Throughout the rest of the paper we will only consider circuit problems where the size of the input circuit $|C|$ is at least $n$ and at most polynomial in $n$ and $m$; this way an algorithm's complexity is equivalent as a function of the number of inputs and outputs or of the actual input size (that is, $|C|$).
\begin{definition}
We say that a circuit problem is circuit-d.s.r.\ \emph{with polynomial-blowup} if  it is circuit d.s.r and furthermore on input $C$ with $n+m$ inputs and outputs, it makes oracle queries on circuits of size at most $|C|+\text{poly}(nm)$.
\end{definition}

The reader might  immediately ask: what is the relation between these notions of self-reducibility? It might seem one is a weaker version of the other, but this does not seem to be the case. In fact, the authors initially believed this was indeed the case, but despite this intuition, attempts at proving it failed! 
For circuit problems in $\TFNP$, it would be interesting if there any connection between these two different notion of downward self-reducibility.

\section{Downward self-reducibility of PLS-complete problems} \label{sec:dsr-problems}

In this section, we initiate the study of downward self-reducibility in $\TFNP$ by showing that of the four (4) canonical $\PLS$-complete problems, three (3) are circuit-d.s.r.\ and one is both circuit-d.s.r.\ and d.s.r. Showing that a circuit problem is d.s.r.\ is trickier, as it requires not only reducing the input and output bits but also \emph{shrinking the entire circuit representation}.
\begin{theorem} \label{thm:ITER}
\textsc{ITER-with-source} is d.s.r.\ and circuit-d.s.r.\ with polynomial blowup.
\end{theorem}

The following definition will be useful for the proof of \cref{thm:ITER}:
\begin{definition}
Given an (acyclic) circuit $C:\B^n\rightarrow \B^m$, the input restriction of $C$ on input bit $i \in [n]$ to value $b \in \B$, denoted $C^{i\rightarrow b}$, is the circuit obtained by setting the input $i$ to $b$, and then removing input bit $i$ by ``propagating" its value through the circuit (the exact process is given following this definition). Similarly, the output restriction of $C$ formed by excluding output bit $j \in [m]$, denoted $C_{\backslash j}$, is defined by removing the $j$-th output bit and propagating the effect throughout the circuit.
\end{definition}

Propagating a restriction is intuitive, but the details are given for completeness. For an input restriction $C^{i\rightarrow b}$, initially, input bit $i$ is replaced by the constant value $b$. In polynomial time the circuit can be partitioned into layers $1,\ldots,L$ such that a gate in later $\ell$ has inputs in layers $\ell'< \ell$ and layer $1$ are the input bits. For each layer $\ell=2,\ldots,L$, we modify every gate $G$ in layer $\ell$ that has as at least one input that was replaced by a constant in an earlier layer:
\begin{itemize}
\item if $G =$ NOT and its input was replaced with constant $\beta$, replace $G$ with constant $\lnot\beta$;
\item if $G \in~$\{AND,OR\} and both input bits $\iota,\kappa$ were replaced with constants, replace $G$ with the value $(\iota~G~\kappa$);
\item if $G =$ AND and only input bit $\iota$ was set to constant $\beta$, if $\beta=0$ replace $G$ by $0$, and if $\beta=1$ remove $G$ by replacing any outgoing wires from $G$ into layers $\ell'>\ell$ with outgoing wires from $\kappa$;
\item if $G =$ OR and only input bit $\iota$ was set to constant $\beta$, if $\beta=1$ replace $G$ by $1$, and if $\beta=0$ remove $G$ by replacing any outgoing wires from $G$ into layers $\ell'>\ell$ with outgoing wires from $\kappa$.
\end{itemize}

For an output restriction $C_{\backslash j}$ we start by removing output bit $j$, and then for each layer $\ell \in \{L-1,\ldots,1\}$, remove any gate which feeds only into gates that were removed in previous layers. Both input and output restrictions can be constructed in polynomial time. Both kinds of restriction have a smaller circuit representation in any reasonable representation of circuits\footnote{We note that is possible to come up with a ``pathological" circuit representation such that restrictions with propagation of this kind result in a longer circuit representation, even though the number of nodes in the DAG encoded by the representation has decreased.}, that is, $|C^{i\rightarrow b}| < |C|$ and $|C_{\backslash j}| < |C|$.

\begin{proof}[Proof of \cref{thm:ITER}] The key observation is that we can first query a restriction of $S$ to the first half of inputs (that is, those beginning with a $0$) and either obtain a sink, or else a string whose successor is in the second half of inputs (those beginning with $1$). With this new source, we can reduce the problem to finding a sink in the second half of inputs.

Concretely, given an instance $(S,s)$ such that $S:\B^n\rightarrow\B^n$ and $S(s) > s$, define $S_0 = (S^{1\rightarrow 0})_{\backslash 1}$ and $S_1=(S^{1\rightarrow 1})_{\backslash 1}$, i.e., the input restriction with the first bit set to $0$ or $1$ respectively, and the output restriction obtained by excluding the first output bit $1$. These restrictions can be constructed efficiently and are necessarily smaller in total representation than $S$, and $S_0(x) = S(0||x)_{1:n-1}$, $S_1(x) = S(1||x)_{1:n-1}$ where the subscript denotes returning that range of bits.

Now, if $s > 2^{n-1}$, then $(S_1,s_{1:n-1})$ defines a valid smaller instance of \textsc{ITER-with-source} and given solution $w$ to $(S_1,s_{1:n-1})$, $v = 1||w$ is a solution to $(S,s)$. Otherwise, assuming $s < 2^{n-1}$, we query the oracle for $(S_0,s_{1:n-1})$. Any solution $w$ of this sub-instance for which $S(0||w) < 2^{n-1}-1$ and $S(S(O||w)) < 2^{n-1}-1$ directly yields a solution for $S$, namely $v=0||w$. Otherwise, there are several possible kinds of ``false" solutions. In each case we will have either a solution to $S$, or otherwise a vertex $\sigma \in \B^n$ such that $\sigma > 2^{n-1}$ such that $S(\sigma) > \sigma$ and in which case we query $(S_1,\sigma_{1:n-1})$ for solution $z$, and output $v = 1||z$.

Again, if $w$ is the solution returned by the downward \textsc{ITER} oracle on $(S_0,s_{1:n-1})$:
\begin{enumerate}
\item If $S(0||w) \geq 2^{n-1}$, then either $S(S(0||w)) < S(0||w)$ (and then $0||w$ is a solution for $(S,s)$) or we set $\sigma = S(0||w)$ and continue as above.
\item If $S(0||w) < 2^{n-1}$ but $S(S(0||w)) \geq 2^{n-1}$, then either $S(S(S(0||w))) < S(S(0||w))$ (in which case $S(0||w)$ is a solution for $(S,s)$), or set $\sigma = S(S(0||w))$.
\end{enumerate}

Note, that in addition to shrinking the circuit size, the oracle sub-instances have one less input and output bit, so \textsc{ITER-with-source} is also circuit-d.s.r.\ with sub-constant blowup.
\end{proof}

\begin{theorem}
\textsc{ITER}, \textsc{Sink-of-DAG} and \textsc{Sink-of-DAG-with-source} are circuit-d.s.r.\ with polynomial blowup.
\end{theorem}

For the notion of circuit-d.s.r.\ to make sense for \textsc{Sink-of-DAG} and \textsc{Sink-of-DAG-with-source}, which as defined have several input circuits with the same domain, we think of the inputs to these problems as a single circuit which reuses the same input bits and concatenates the outputs. This works because although these problems are defined with multiple circuit inputs, a solution is a single input (that satisfies some relation depending on the circuits). That is, to be circuit-d.s.r.\ the oracle sub-instances must decrease the input size for \emph{each} circuit, or the output size for \emph{any} circuit.

\begin{proof}

\begin{enumerate}
\item \textbf{ITER}: We follow roughly the same algorithm as in the proof of \cref{thm:ITER} but with $s=0^n$; when we query $S_0 = (S^{1\rightarrow 0})_{\backslash 1}$ we either obtain a solution to the original instance, or a vertex $\sigma \in \B^n$ such that $\sigma \geq 2^{n-1}$ and $S(\sigma) > \sigma$. The problem is that we cannot just input-restrict the first bit to $1$, since the source of this new circuit must be $0^{n-1}$ and $S_1(0^{n-1}) = S(2^{n-1})$ and it may not be that $\sigma = 2^{n-1}$; therefore, we modify $S_1 = (S^{1\rightarrow 1})_{\backslash 1}$ to ``redirect" the input $0^{n-1}$ to the new source by adding a component that checks if the input is $0^{n-1}$, and if so feeds a hardcoded $\sigma_{1:n-1}$ into the rest of the circuit, and otherwise feeds in the input unchanged. This gadget can be implemented with $O(n)$ additional gates, so \textsc{ITER} is circuit-d.s.r with polynomial-blowup. 

\item \textbf{Sink-of-DAG-with-source}: Let $V' = V_{\backslash 1}$ be the output restriction of the first bit; we query $(S,V')$ and obtain solution $w$. If both $V(w) < 2^{m-1}$ and $V(S(w)) < 2^{m-1}$ then $w$ is a solution to the original instance; otherwise, this gives us a $\sigma \in \B^n$ such that $V(\sigma) \geq 2^{m-1}$ and $S(\sigma) \neq \sigma$. For the next subquery, the valuation circuit will again be $V'$, but we also use a modified successor circuit $S'$: on input $x\in \B^n$ it checks whether $V(x) < V(\sigma)$ and if so outputs $x$, and otherwise outputs $S(x)$. We also set $s' = \sigma$. The modification to $S'$ ensures no $x$ for which $V(x) < V(\sigma)$ can be a solution to $(S',V',s')$. Since this algorithm only ever decreases the \emph{output} bits of of the valuation circuit in oracle calls, it is necessary to check that it can be solved when output bits cannot be decreased any further, that is when $m=1$; indeed, in that case the problem is trivial can be found by applying $S$ just once. $S'$ needs to be implemented carefully to ensure polynomial-blowup: instead of copying $V$, the output bits of $V$ are moved down into a non-output layer, and $S'$ compares these bits to the hardcoded $S(\sigma)$ ($n-1$ of these bits are also reused directly to compute $V'$). This gadget can be implemented with $O(m)$ additional overhead.

\item \textbf{Sink-of-DAG}: The construction is exactly as above, except $S'$ has an additional component which checks if $x=0^n$, and if so feeds a harcoded $\sigma$ into the rest of the circuit for $S'$ and $V'$. This gadget can be implemented with $O(n)$ additional overhead.\qedhere
\end{enumerate}
\end{proof}

A natural question is whether every $\PLS$-complete problem is d.s.r.\ (or circuit d.s.r.\ if it is a circuit problem). The difficulty with answering this question is that, as evidenced in this section, downward self-reducibility seems to depend on the details of the input representation of the problem. This is analogous to the case of downward self-reducibility in decision problems: SAT, and more generally TQBF, are known to be d.s.r., but the property is not known to be true for \emph{all} $\NP$- or $\PSPACE$-complete problems. The same ideas to show any NP-complete problem has a search-to-decision self-reduction does not work when considering \emph{downward} self-reducibility, either in the decisional landscape or in $\TFNP$. If, say the reduction from SAT to $A$ squares the input length, it is not clear how to use downward self-reducibility of SAT to get it for $A$. 

\begin{openproblem}
Is every $\PLS$-complete problem downward self-reducible?
\end{openproblem}

\section{Downward self-reducible problems are in \PLS}\label{sec:dsrinPLS}

In the previous section we proved that the canonical $\PLS$-complete problems are downward self-reducible. The natural question is whether the reverse is true. Indeed, we are able to show the following theorem (restated from the introduction): 

\dsrinPLS*

The ``moral" of these results (\cref{thm:ITER,thm:dsr-in-PLS} that some $\PLS$-complete problems are d.s.r.\ and d.s.r.\ problems are in $\PLS$) is the following corollary. An equivalent one exists in the decisional landscape, namely, that a hard language is d.s.r.\ iff there is a hard language in $\PSPACE$.

\begin{corollary}  \label{cor:dsr}
A problem in $\PLS$ is hard if and only if a downward-self reducible problem in $\TFNP$ is hard. Here, ``hard" can mean worst-case hard, or average-case hard with respect to some efficiently sampleable distribution.
\end{corollary}
\begin{proof}
This is essentially restatement of \cref{thm:ITER,thm:dsr-in-PLS} together. On one hand, if a problem in $\PLS$ is hard (either worst-case, or average-case with respect to some sampleable distribution) then any $\PLS$ complete problem is also hard (with the same notion of hardness carrying over); in particular, \textsc{ITER}~$\in \PLS$ is hard. Conversely if a d.s.r.\ problem in $\TFNP$ is hard, then by \cref{thm:dsr-in-PLS} the same problem is in $\PLS$.
\end{proof}

\begin{proof}[Proof of \cref{thm:dsr-in-PLS}]
If $R$ is downward self-reducible with algorithm $A'$, then there exists the following (inefficient) depth-first recursive algorithm $A$ for $R$: on input $x$, simulate $A'$ and for any oracle call to an instance $x'$ where $|x'| < x$, run $A(x')$ to obtain a solution, and then continue the simulation of $A'$. The idea of reducing $R$ to \textsc{Sink-of-DAG} is to represent its intermediate states of this algorithm as vertices of a DAG. Of course, this graph may have exponential size.

Intuitively, we will represent the intermediate states of $A$ as a list of the recursive instances invoked by $A$ up to the current depth of recursion, along with the solutions to already-solved recursive instances. We can assume without loss of generality that on inputs of length $n$, $A$ makes exactly $p(n)$ downward calls (this can be done by ``padding" the algorithm with unused additional calls to some fixed sub-instance) and that solutions have size $q(n)$.

The \textsc{Sink-of-DAG} vertex set will consist of strings $s \in \B^{P(n)}$ where $P(n) = n + q(n) + \sum_{i=1}^{n-1} p(n-i+1) \times (n-i+q(n-i))$ which we will call \emph{states} or \emph{nodes}. Note $P(n) < p(n)q(n)n^2$ and is hence polynomial size. We interpret a string $s \in \B^{P(n)}$ as a table $s[\cdot,\cdot]$ indexed by the depth of recursion $i \in \{0,\ldots,n-1\}$, and index $j$ of the recursive call at depth $i$, with $j \in [p(n-i+1)]$ (except for $i=0$ where $j$ can only be $1$). Each $s[i,j] \in \B^{n-i}\times B^{q(n-i)}$ can contain an instance of $R$, and possibly a solution. That is, $s[i,j]$ can be one of: (1) a pair $(\xi,\bot)$ of an $R$-instance $\xi$ and special symbol $\bot$, (2) a pair $(\xi,\gamma)$ of an $R$-instance $\xi$ and purported solution $\gamma$, or (3) the pair $(\bot,\bot)$ representing no instance.

\begin{algorithm}[h]
\caption{\textsc{Valid?} (Algorithm for whether state $s$ is valid in \cref{thm:dsr-in-PLS})}
\label{alg:valid}
\KwIn{\begin{itemize}
    \item State $s \in \B^{P(n)}$, $s[i,j] \in \B^{n-i}\times\B^{q(n-1)}$ for $i \in [n-1]$ and $j \in [p(n-i+1)]$.
    \item Input $x$ to search problem $R$. 
    \end{itemize}
}
\KwOut{\textsc{True} if state $s$ is valid for input $x$, else \textsc{False}}
\hrulefill \\
Set $(\xi,y) \gets  s[0,1]$\\
\lIf{$\xi \neq x$}
    {\KwRet{\textsc{false}}}
\lIf{$y \neq \bot$}
    {\KwRet{"$(x,y) \stackrel{?}{\in} R$"}}
Set $j \gets 1$\\
\While{$(x_j,y_j) := s[1,j] \neq (\bot,\bot)$}{
    \lIf{$x_j$ is not the next input queried by $A$ after $x_1,\ldots,x_{j-1}$ with solutions $y_1,\ldots,y_{j-1}$}
        {\KwRet{\textsc{False}}}
    \lIf{$y_j \neq \bot$ \bf{and} $(x_j,y_j) \notin R$ }
        {\KwRet{\textsc{False}}}
    Set $s_{x_j} \in \B^{P(n-1)}$ as follows: $s_{x_j}[0,1] = s[1,j]$, and $s_{x_j}[\iota,\kappa] = s[\iota+1,\kappa]$ for $\iota \in [n-1]$. \\
    \lIf{\textsc{Valid?}$(s_{x_j},x_j)$}
        {break while loop}
    \lElse {\KwRet{\textsc{False}}}
    Set $j \gets j+1$
    }
\For{$j'\in \{j+1,\ldots,p(n)\}$}{
    \lIf{$s[1,j']\neq(\bot,\bot)$}
        {\KwRet{\textsc{False}}}
    }
\KwRet{\textsc{True}}.

\end{algorithm}

On input $x$, the starting state $s_0$ is defined by setting $s_0[0,1] = (x,\bot)$ and $s[i,j] = (\bot,\bot)$ for all other $i,j$. We say that a state $s$ is \emph{valid} for input $x$ based on the following recursive definition: $s[0,1] = (x,\bot)$, and there exists $j \in [p(n)]$ (where $n=|x|$) such that 
\begin{enumerate}
\item $s[1,j'] = (\bot,\bot)$ for all $j' > j$,
\item for $j' \leq j$, $s[1,j'] = (x_{j'},y_{j'})$ where $y_{j'}$ is a correct solution for $x_{j'}$ (except for $j' = j$ where $y_{j'}$ can be $\bot$),
\item for each $j' \leq j$, $x_{j'}$ is the next $(n-1)$-size subquery made by $A$ on input $x$ given previous subqueries $x_{1},\ldots,x_{j'-1}$ and  respective solutions $y_{1},\ldots,y_{j'-1}$,
\item if $y_j = \bot$, then letting $s_{x_j} \overset{\Delta}{=} s[i',j']$ be the table defined by $s_{x_j}[0,1]=(x_j,\bot)$ and $s[i',j'] = s[i'+1,j']$ for $i' \in [n-2]$ and $j'\in p((n-1)-i'+1)$, $s_{y_j}$ is a valid state for input $x_j$.
\end{enumerate} 
See \cref{alg:valid} for the pseudocode for checking if a state is valid.

Condition 4 says that, letting $s_{x_j}$ be the subtable of $s$ corresponding to the computation of $A$ to solve the instance $x_j$-- that is, $(x_j,\bot)$ in its first cell $s_{x_j}[0,1]$, and the rest of $s[\iota,j]$ for $\iota > 2$-- $s_{x_j}$ is valid.

Validity can be checked efficiently by the recursive algorithm implicit to this definition; its time complexity is $T(n) \leq \text{poly}(p(n)) + T(n-1) = \text{poly}(n)$.

\begin{algorithm}[h]
\caption{$S$ (Algorithm for the successor circuit $S$ in \cref{thm:dsr-in-PLS})}
\label{alg:successor}
\KwIn{\begin{itemize}
    \item State $s \in \B^{P(n)}$ where $s[i,j] \in \B^{n-i}\times\B^{q(n-1)}$ for $i \in [n-1]$ and $j \in [p(n-i+1)]$.
    \item Input $x$ to search problem $R$, and $n := |x|$.
    \end{itemize}
}
\KwOut{Successor state of state $x$ on input $x$}
\hrulefill \\
\lIf{\textsc{valid?}$(s,x)=$\textsc{False}}
    {\KwRet{$s$}}
Set $(x,y) \gets s[0,1]$\\
\lIf{$y\neq \bot$}{\KwRet{$s$}}
Let $j \in [p(n)]$ such that $\forall~j'>j$, $s[1,j']=(\bot,\bot)$\\
Set $(x_j,y_j) \gets s[1,j]$\\
\uIf{$j = p(n)$ \bf{and} $y_j \neq \bot$}
    {Simulate $A(x)$ after queries $x_1,\ldots,x_j$ with solutions $y_1,\ldots,y_j$ to obtain $y$ such that $(x,y) \in R$.\\
    Let $s' \in \B^{P(n)}$ with $s'[0,1]=(x,y)$ and $(\bot,\bot)$ elsewhere.\\
    \KwRet{$s'$}.
    }
\uElseIf{$y_j\neq \bot$}
    {Simulate $A(x)$ after queries $x_1,\ldots,x_j$ with solutions $y_1,\ldots,y_j$ to obtain next query $x_{j+1}$.\\
    Let $s' := s$, but $s'[1,j+1]:=(x_{j+1},\bot)$. \\
    \KwRet{$s'$}.
    }
\Else{
    Let $s_{x_j} \subset s$ be the subset of cells of $s$ as indexed by $s_{x_j}[0,1] = s[1,j]$, and $s_{x_j}[\iota,\kappa] = s[\iota+1,\kappa]$ for $\iota \in [n-1]$. \\
    Let $s_{x_j}' = S(s_{x_j},x_j)$.\\
    Let $s':=s$ with subrange $s_{x_j}$ replaced by $s_{x_j}'$.\\
    \KwRet{$s'$}.
    }
\end{algorithm}

The successor function $S$ is also defined recursively. For invalid states the successor function returns the state itself; as they cannot be reached this makes them isolated nodes that cannot be a solution to the \textsc{Sink-of-Dag} instance. For valid states, let $(x,y) = s[0,1]$. If $y\neq \bot$ is a solution, then $S(s) := s$ (the node is a sink as it contains a solution to $x$). Otherwise, if $y=\bot$ (and $s$ represents an intermediate state of the algorithm $A$), let $j \in p[n]$ be the index of the last subquery (of size $n-1$) encoded by $s$, exactly as in conditions (2)-(4) in the algorithm for validity, with corresponding cell $(x_j,y_j) = s[1,j]$. If $y_j \neq \bot$ and $j < p(n)$, we simulate $A$ to obtain the next subquery $x_{j+1}$ of size $n-1$ (given the previous subqueries $x_1,\ldots,x_j$ with solutions $y_1,\ldots,y_j$) and set the $[1,j+1]$-cell to $(x_{j+1},\bot)$ (that is, $S(s)[1,j+1] = (x_{j+1},\bot)$ and $S(s)[i',j'] = s[i',j']$ for all other $(i',j')\neq (1,j+1)$). Note that this part is non-recursive (since by definition, $x_{j+1}$ is the very next recursive call made by $A$) and requires polynomial time. If $y_j \neq \bot$ and $j = p(n)$, we similarly simulate $A$ to obtain the final solution $y$ to $x$ since all the subqueries made by $A$ on input $x$ have been obtained with their answers; we set $S(s)[1,0] = (x,y)$ and set all other cells to $(\bot,\bot)$. Finally if $y_j = \bot$, $s$ represents an intermediate state of $A$ currently at a deeper level of recursion. Construct $s_{x_j}$ as in the validity algorithm ($s_{x_j}$ is a subset of the cells of $s$, namely $(x_j,\bot)$ as cell $[0,1]$, and the rest of $s$ at levels $2$ and beyond); let $\sigma_{x_j}=S(s_{x_j})$ be the successor of $s_{x_j}$ when simulating $A$ on $s_{x_j}$; $S(s)$ returns $s$, replacing the subset of cells corresponding to $s_{x_j}$ with $\sigma_{x_j}$. Observe that the time to compute the successor is at most the time to simulate $A$ between recursive calls or the time to compute the successor for a smaller sub-instance, that is $T_s(n) \leq \text{poly}(n) + T_s(n-1) = \text{poly}(n)$. See \cref{alg:successor} for the pseudocode for checking if a state is valid.

Finally, for valid nodes, the potential circuit $V$ will return how far the state is in the depth-first recursive algorithm. The position $\pi(s) \in \B^{cn^2}$ of a valid state $s$ in the simulation of $A$ can be computed as follows. Let $j_i$ be the last cell at each level $i=1,\ldots,n-1$ that is not $(\bot,\bot)$ (i.e., $j_i$ is the last index such that $s[i,j_i]$ contains a subquery). If $s$ is a solution (sink), $\pi(s) =\Pi(n)$, and otherwise $\pi(s) = 1 + \sum_{i=1}^{n-1} \min\{1,j_i\}+\max\{0,j_i-1\}\times\Pi(n-i)$.

It is clear that the only valid node to not have a valid successor is the one corresponding to the final configuration of $A$, with solution $y$, a solution to the \textsc{Sink-of-DAG} instance gives a solution to $x$ (by reading $(x,y) = s[0,1]$).
\end{proof}

In \cref{sec:dsr,sec:dsr-problems} we discussed the related notion of circuit-d.s.r.\ in which the sub-instances to the oracle need not be smaller, but the number of input or output bits decreases. While we do not know the exact connection between these notions, in particular whether they are equivalent for circuit problems, the following theorem establishes that it is enough for a $\TFNP$ problem to be downward self-reducible in either sense to guarantee membership in $\PLS$.

\begin{theorem} \label{thm:cdsr-pls} Let $R$ be a circuit problem in $\TFNP$ that is circuit-d.s.r.\ with polynomial blowup. Then $R\in\PLS$.
\end{theorem}
\begin{proof}
Let $A'$ be the algorithm from the definition of circuit-d.s.r. As in d.s.r., $A'$ can be used to construct an (inefficient) depth-first recursive algorithm $A$ that solves $R$: simulate $A'(C)$, and on oracle calls to a sub-instance $C'$, run $A(C')$. Critically, $A$ terminates even though $|C'|$ may be greater than $|C|$, because the sum of input and output bits $n+m$ decreases with each level of recursion.

The reduction to \textsc{Sink-of-DAG} is the same as in the proof of \cref{thm:dsr-in-PLS}: the graph will represent intermediate states of the depth-first recursive algorithm $A$. However, the size of states will be different. Without loss of generality assume $R$ makes exactly $p(|C|)$ sub-instance calls, and if $R$ is circuit-d.s.r.\ with polynomial-blowup, we have that each sub-instance has size at most $|C|+(nm)^c$ for some $c > 0$. Hence, for a circuit of size $|C|$ we have $P(|C|) = |C|+ n + \sum_{i=1}^{n+m-1} p(|C|+i(nm)^c)(|C|+i(nm)^c)n \leq \text{poly}(n,m)$ bits suffices to represent intermediate states. 
\end{proof}

Note that the size of representing intermediate states of $A$ is where it is critical that $R$ be circuit-d.s.r with polynomial blowup; if instead, for instance, we had the guarantee that sub-instances have size $f(|C|)$, then at depth-$i$ sub-instances would have size $f^{(i)}(|C|)$ which even for the modest function $f(\alpha) = 2\alpha$ would give sub-instances of size $2^i |C|$, which are too big to represent as \textsc{Sink-of-DAG} nodes. 

Finally, we can give a characterization of $\PLS$-hardness analogous to \cref{cor:dsr}:
\begin{corollary}
A problem in $\PLS$ is hard if and only if there exists a hard circuit problem in $\TFNP$ that is circuit-d.s.r.\ with polynomial blowup.
\end{corollary}
\begin{proof}
This is a restatement of \cref{thm:cdsr-pls,thm:cdsr-problems} together.
\end{proof}

\section{Unique TFNP and Downward self-reducibility}\label{sec:tfup}

Several highly structured problems in $\TFNP$ have the additional proprety that for every input, the solution is \emph{unique}. Most of these problems are algebraic in nature-- such as factoring, and finding the discrete logarithm of an element of a cyclic group for a certain generator. We find, perhaps surprisingly, that d.s.r.\ problems with the additional condition of unique solutions are not only in $\PLS$, but in $\CLS$! 

\begin{definition}[\cite{HubacekV2021}]
\textsc{TFUP}, for ``unique" $\TFNP$, is the class of problems in $\TFNP$ where for each input $x$ there is exactly one solution $y$.
\end{definition}


\begin{definition} \label{def:svl}
\textsc{Sink-of-Verifiable-Line} (SVL): given a successor circuit $S:\B^n\rightarrow \B^n$, source $s\in \B^n$, target index $T \in [2^n]$ and verifier circuit $V:\B^n\times [T] \rightarrow \B$ with the guarantee that for $(x,i) \in \B^n\times[T]$, $V(x,i) = 1$ iff $x = S^{i-1}(s)$, find the string $v \in \B^n$ s.t. $V(v,T)=1$.
\end{definition}

SVL is a \emph{promise} problem; in particular, there is no way to efficiently check the guarantee about the verifier circuit $V$. A solution is guaranteed to exist if the promise is true, and the solution must be unique. It is possible to modify the problem (by allowing solutions that witness an ``error" of the verifier circuit) in order to obtain a search problem in $\TFNP$, but for our purposes SVL is sufficient:


\begin{theorem}[\cite{FearnleyGMS2020}] \label{thm:svl-cls}
SVL is polynomial-time reducible to $\UEOPL$.
\end{theorem}

Bitansky, Paneth and Rosen \cite{BitanskyPR2015}, building on the results of Abbot, Kane and Valiant \cite{AbbotKV2004}, used a reversible pebbling game (also known as the east model) to make the computation of the successor circuit reversible and show that SVL is polynomial-time reducible to $\PPAD$ (see \cite[Section 4.3]{AbbotKV2004} and \cite[Section 3]{BitanskyPR2015}). Huba\'{c}ek and Yogev \cite{HubacekY2020} then strengthened this result by showing that the same reduction shows SVL is polynomial-time reducible to $\CLS$. Finally, in recent work this was strengthened by observing that the same reduction works for $\UEOPL$ \cite{FearnleyGMS2020}.


\begin{lemma}
Every d.s.r.\ problem and every circuit-d.s.r.\ circuit problem in $\TFUP$ reduces to SVL.
\end{lemma}
\begin{proof}
Given a d.s.r.\ (or circuit-d.s.r.) problem $R$, we construct a graph with successor circuit $S$ exactly as in the proof of \cref{thm:dsr-in-PLS} or \cref{thm:dsr-in-PLS}. The key observation is that if $R$ has unique solutions, we can construct a verifier circuit $V$ as in \cref{def:svl} for which the SVL promise holds true.

Let $\Pi(n)$ denote the total length of the path from the initial state $s_0$ to the unique sink in the \textsc{Sink-of-DAG} instance from the earlier theorem proofs; we have $\Pi(n) = p(n) \times \Pi(n-1) = p(n) \times p(n-1) \cdots \times p(1)$. While $\Pi(n)$ may be exponential it can be represented with $\sum \log(p(n-i)) = O(n^2)$ bits; that is, the distance of each valid node from $s_0$ in the graph can be represented with $cn^2$ for some constant $c$. The position $\pi(s) \in \B^{cn^2}$ of a valid state $s$ in the simulation of $A$ can be computed as follows. Let $j_i$ be the last cell at each level $i=1,\ldots,n-1$ that is not $(\bot,\bot)$ (i.e., $j_i$ is the last index such that $s[i,j_i]$ contains a subquery). If $s$ is a solution (sink), $\pi(s) =\Pi(n)$, and otherwise $\pi(s) = 1 + \sum_{i=1}^{n-1} \min\{1,j_i\}+\max\{0,j_i-1\}\times\Pi(n-i)$. Equivalently this can be computed recursively; using the sub-table notation from \cref{thm:dsr-in-PLS}, where $j$ is the last cell at level $1$ that is not $(\bot,\bot)$, $\pi(s) = 1+(j-1)\Pi(n-1)+\pi(s_{x_j})$. Either way, $\Pi$ and $\pi$ can be computed efficiently. Since the sink is the $\Pi(n)$-th state, the SVL instance will have target $T=\Pi(n)$.

Critically, since $R \in \TFUP$, for each $i \in [\Pi(n)]$ there is exactly one state $s\in \B^{P(n)}$ such that $\pi(s) = i$: to see this, for every sub-instance $x'$ invoked by $A$, there is a unique solution $y'$, and given a sequence of previous sub-instances and solutions at a given level of recursion, the next sub-instance is uniquely determined by $A$. We define $V(s,i) = 1$ if $s$ is a valid state (as in the proof of \cref{thm:dsr-in-PLS}) and $\pi(s) = i$.
\end{proof}

As a consequence of \cref{thm:svl-cls} and the previous lemma, we have the following theorem (restated from the introduction).

\tfup*

This result yields an interesting application to the study of number-theoretic total problems such as factoring. It is an open problem whether there exists an algorithm for factoring that uses oracle calls on smaller numbers. The following observations are evidence that the answer may be no. 

\begin{definition}
\textsc{AllFactors} is the problem of, given an integer, listing its non-trivial factors in increasing order, or ``prime" if it is prime. \textsc{Factor} is the problem of returning \emph{a} non-trivial factor, or ``prime" if it is prime.
\end{definition}

\textsc{AllFactors} is in $\TFUP$. \textsc{AllFactors} and \textsc{Factor} are almost the same problem: \textsc{Factor} is Cook reducible to \textsc{AllFactors} with one oracle call and \textsc{AllFactors} is Cook reducible to \textsc{Factor} with $\log$(input size) oracle calls. A consequence of \cref{thm:tfup-cls} is the following:

\factor*

\begin{proof}
If either problem is downward self-reducible, so is the other. Given a downward self-reduction $A$ of \textsc{Factor}, to solve \textsc{AllFactors} on input $n$, run $A$ (replacing its downward oracle calls with those for \textsc{AllFactors}) to obtain factor $m$, and then run the oracle on $\frac{n}{m}$. Given a downward self-reduction $B$ of \textsc{AllFactors}, to solve \textsc{Factor} on input $n$ run $B$ and return any factor-- replace downward oracle calls to \textsc{AllFactors} with ones to \textsc{Factor} and at most $\log(\log(n))$ additional downward queries. Hence if either problem is d.s.r., \textsc{AllFactors} $\in \UEOPL$, but since \textsc{Factor} $\leq_p$ \textsc{AllFactors}, \textsc{Factor} $\in \UEOPL$.
\end{proof}

\textsc{Factor} $\in \UEOPL$, and hence in $\CLS$, would be a surprising consequence to $\TFNP$ community. While it has been shown that \textsc{Factor} can be reduced to the $\TFNP$-subclasses $\textsc{PPA}$ and $\PPP$, the intersection of which contains $\PPAD$, it has been a long outstanding question whether \textsc{Factor} $\overset{?}{\in} \PPAD$. An algorithm for factoring $N$ which queries numbers at most $N/2$ (this is what it means for the input size to shrink by at least one bit, but note that any number that could contain a non-trivial factor of $N$ is at most $N/2$) would immediately place \textsc{Factor} not just in $\PPAD$, but also in $\PLS$ (since $\CLS = \PPAD~\cap~\PLS$)!

\section{Discussion and Open Problems}\label{sec:disc}
In this paper we have initiated the study of downward self-reducibility, and more broadly self-reducibility, in $\TFNP$. Naturally, a host of questions remain:
\begin{itemize}
\item What is the relationship between downward self-reducibility and circuit-downward self reducibility for circuit problems in $\TFNP$?
\item What other problems in $\PLS$ are downward self-reducible? Is it possible every $\PLS$-complete problem is downward self-reducible? As a first step, it would be helpful to prove that the suite of $\PLS$-complete local constraint/clause maximization problems (where a specific neighborhood function is specified as part of the problem, such as \textsc{Flip} or \textsc{Kernighan-Lin}) are d.s.r.
\item Another important notion of self-reducibility is that of {\em random} self-reducibility (r.s.r.): a problem $R$ is r.s.r. if it has a worst-to-average case reduction. The details vary but one definition is $R$ is r.s.r. if it has a randomized Las Vegas algorithm if there exists an algorithm solving $R$ on $1/\text{poly}(n)$-fraction of inputs. Some very important algebraic problems in $\TFNP$, such as the discrete logarithm, and RSA inversion, are r.s.r. These problems are also in $\PPP$, so a natural starting point is to ask whether all r.s.r. problems in TFNP are in $\PPP$.
\end{itemize}

\ifitcs
\bibliographystyle{myplainurl}
\else 
{\small
\bibliographystyle{prahladhurl.bst}
\fi 

\bibliography{HMR-bib}

\begin{thebibliography}{GHJMPRT22}

\bibitem[AKV04]{AbbotKV2004}
\textsc{Tim Abbot}, \textsc{Daniel Kane}, and \textsc{Paul Valiant}.
\newblock \href{http://web.mit.edu/tabbott/Public/final.pdf} {\emph{On
  algorithms for {N}ash equilibria}}, 2004.
\newblock (manuscript).

\bibitem[All10]{Allender2010}
\textsc{Eric Allender}.
\newblock \href{https://people.cs.rutgers.edu/~allender/papers/cie.plenary.pdf}
  {\emph{New surprises from self-reducibility}}.
\newblock In \textsc{Fernando Ferreira}, \textsc{Helia Guerra}, \textsc{Elvira
  Mayordomo}, and \textsc{João Rasga}, eds., \emph{Programs, Proofs, Processes
  (Abstract and Handout Booklet, 6th Conference on Computability in Europe
  (CiE))}, pages 1--5. Centre for Applied Mathematics and Information
  Technology, Dept. of Mathematics, University of Azores, Portugal, 2010.

\bibitem[Bal90]{Balcazar1990}
\textsc{Jos{\'{e}}~L. Balc{\'{a}}zar}.
\newblock \href{https://doi.org/10.1016/0022-0000(90)90025-G}
  {\emph{Self-reducibility}}.
\newblock J. Comput.\ Syst.\ Sci., 41(3):367--388, 1990.

\bibitem[BCHKLPR22]{BitanskyCHKLPR22}
\textsc{Nir Bitansky}, \textsc{Arka~Rai Choudhuri}, \textsc{Justin Holmgren},
  \textsc{Chethan Kamath}, \textsc{Alex Lombardi}, \textsc{Omer Paneth}, and
  \textsc{Ron~D. Rothblum}.
\newblock \emph{{PPAD} is as hard as {LWE} and iterated squaring}.
\newblock In \textsc{Eike Kiltz} and \textsc{Vaikuntanathan Vinod}, eds.,
  \emph{Proc.\ $20$th International Theory of Crypt.\ Conf.\ (TCC)}, volume
  13747 of \emph{LNCS}. Springer, 2022.
\newblock (to appear).
\newblock
  \href{https://eprint.iacr.org/2022/1272}{\path{eprint.iacr:2022/1272}}.

\bibitem[BG20]{BitanskyG2020}
\textsc{Nir Bitansky} and \textsc{Idan Gerichter}.
\newblock \href{https://doi.org/10.4230/LIPIcs.ITCS.2020.6} {\emph{On the
  cryptographic hardness of local search}}.
\newblock In \textsc{Thomas Vidick}, ed., \emph{Proc.\ $11$th Innovations in
  Theor.\ Comput.\ Sci.\ (ITCS)}, volume 151 of \emph{LIPIcs}, pages 6:1--6:29.
  Schloss Dagstuhl, 2020.
\newblock \href{https://eprint.iacr.org/2020/013}{\path{eprint.iacr:2020/013}}.

\bibitem[BO06]{Buresh-Oppenheim2006}
\textsc{Joshua Buresh-Oppenheim}.
\newblock \href{http://www.cs.toronto.edu/~bureshop/factor.pdf} {\emph{On the
  {TFNP} complexity of factoring}}, 2006.
\newblock (manuscript).

\bibitem[BPR15]{BitanskyPR2015}
\textsc{Nir Bitansky}, \textsc{Omer Paneth}, and \textsc{Alon Rosen}.
\newblock \href{https://doi.org/10.1109/FOCS.2015.94} {\emph{On the
  cryptographic hardness of finding a {N}ash equilibrium}}.
\newblock In \textsc{Venkatesan Guruswami}, ed., \emph{Proc.\ $56$th IEEE
  Symp.\ on Foundations of Comp.\ Science (FOCS)}, pages 1480--1498. 2015.
\newblock
  \href{https://eccc.weizmann.ac.il/eccc-reports/2015/TR15-001}{\path{eccc:2015/TR15-001}},
  \href{https://eprint.iacr.org/2014/1029}{\path{eprint.iacr:2014/1029}}.

\bibitem[CHKPRR19a]{ChoudhuriHKPRR2019-fs}
\textsc{Arka~Rai Choudhuri}, \textsc{Pavel Hub{\'{a}}cek}, \textsc{Chethan
  Kamath}, \textsc{Krzysztof Pietrzak}, \textsc{Alon Rosen}, and \textsc{Guy~N.
  Rothblum}.
\newblock \href{https://doi.org/10.1145/3313276.3316400} {\emph{Finding a
  {N}ash equilibrium is no easier than breaking {F}iat-{S}hamir}}.
\newblock In \textsc{Moses Charikar} and \textsc{Edith Cohen}, eds.,
  \emph{Proc.\ $51$st ACM Symp.\ on Theory of Computing (STOC)}, pages
  1103--1114. 2019.
\newblock
  \href{https://eccc.weizmann.ac.il/eccc-reports/2019/TR19-074}{\path{eccc:2019/TR19-074}},
  \href{https://eprint.iacr.org/2019/549}{\path{eprint.iacr:2019/549}}.

\bibitem[CHKPRR19b]{ChoudhuriHKPRR2019-ppad}
---{}---{}---.
\newblock \emph{{PPAD}-hardness via iterated squaring modulo a composite},
  2019.
\newblock (manuscript).
\newblock \href{https://eprint.iacr.org/2019/667}{\path{eprint.iacr:2019/667}}.

\bibitem[DP11]{DaskalakisP2011}
\textsc{Constantinos Daskalakis} and \textsc{Christos~H. Papadimitriou}.
\newblock \href{https://doi.org/10.1137/1.9781611973082.62} {\emph{Continuous
  local search}}.
\newblock In \textsc{Dana Randall}, ed., \emph{Proc.\ $22$nd Annual
  {ACM}-{SIAM} Symp.\ on Discrete Algorithms (SODA)}, pages 790--804. 2011.

\bibitem[EFKP20]{EphraimFKP2020}
\textsc{Naomi Ephraim}, \textsc{Cody Freitag}, \textsc{Ilan Komargodski}, and
  \textsc{Rafael Pass}.
\newblock \href{https://doi.org/10.1007/978-3-030-45727-3\_5} {\emph{Continuous
  verifiable delay functions}}.
\newblock In \textsc{Anne Canteaut} and \textsc{Yuval Ishai}, eds.,
  \emph{Proc.\ $39$th Annual International Conf.\ on the Theory and Appl.\ of
  Cryptographic Tech (EUROCRYPT), Part {III}}, volume 12107 of \emph{LNCS},
  pages 125--154. Springer, 2020.
\newblock \href{https://eprint.iacr.org/2019/619}{\path{eprint.iacr:2019/619}}.

\bibitem[FGHS21]{FearnleyGHS2021}
\textsc{John Fearnley}, \textsc{Paul~W. Goldberg}, \textsc{Alexandros
  Hollender}, and \textsc{Rahul Savani}.
\newblock \href{https://doi.org/10.1145/3406325.3451052} {\emph{The complexity
  of gradient descent: {CLS} = {PPAD} {\(\cap\)} {PLS}}}.
\newblock In \textsc{Samir Khuller} and \textsc{Virginia~Vassilevska Williams},
  eds., \emph{Proc.\ $53$rd ACM Symp.\ on Theory of Computing (STOC)}, pages
  46--59. 2021.
\newblock \href{http://arxiv.org/abs/2011.01929}{\path{arXiv:2011.01929}}.

\bibitem[FGMS20]{FearnleyGMS2020}
\textsc{John Fearnley}, \textsc{Spencer Gordon}, \textsc{Ruta Mehta}, and
  \textsc{Rahul Savani}.
\newblock \href{https://doi.org/10.1016/j.jcss.2020.05.007} {\emph{Unique end
  of potential line}}.
\newblock J. Comput.\ Syst.\ Sci., 114:1--35, 2020.
\newblock (Preliminary version in \emph{46th ICALP}, 2019).
\newblock \href{http://arxiv.org/abs/1811.03841}{\path{arXiv:1811.03841}}.

\bibitem[GHJMPRT22]{GoosH0MPRT2022}
\textsc{Mika G{\"{o}}{\"{o}}s}, \textsc{Alexandros Hollender},
  \textsc{Siddhartha Jain}, \textsc{Gilbert Maystre}, \textsc{William Pires},
  \textsc{Robert Robere}, and \textsc{Ran Tao}.
\newblock \href{https://doi.org/10.4230/LIPIcs.CCC.2022.33} {\emph{Further
  collapses in {TFNP}}}.
\newblock In \textsc{Shachar Lovett}, ed., \emph{Proc.\ $37$th Comput.\
  Complexity Conf.}, volume 234 of \emph{LIPIcs}, pages 33:1--33:15. Schloss
  Dagstuhl, 2022.
\newblock \href{http://arxiv.org/abs/2202.07761}{\path{arXiv:2202.07761}},
  \href{https://eccc.weizmann.ac.il/eccc-reports/2022/TR22-018}{\path{eccc:2022/TR22-018}}.

\bibitem[GPS16]{GargPS2016}
\textsc{Sanjam Garg}, \textsc{Omkant Pandey}, and \textsc{Akshayaram
  Srinivasan}.
\newblock \href{https://doi.org/10.1007/978-3-662-53008-5\_20}
  {\emph{Revisiting the cryptographic hardness of finding a {N}ash
  equilibrium}}.
\newblock In \textsc{Matthew Robshaw} and \textsc{Jonathan Katz}, eds.,
  \emph{Proc.\ $36$th Annual International Cryptology Conf.\ (CRYPTO), Part
  {II}}, volume 9815 of \emph{LNCS}, pages 579--604. Springer, 2016.
\newblock
  \href{https://eprint.iacr.org/2015/1078}{\path{eprint.iacr:2015/1078}}.

\bibitem[HV21]{HubacekV2021}
\textsc{Pavel Hub{\'{a}}cek} and \textsc{Jan V{\'{a}}clavek}.
\newblock \href{https://doi.org/10.4230/LIPIcs.MFCS.2021.60} {\emph{On search
  complexity of discrete logarithm}}.
\newblock In \textsc{Filippo Bonchi} and \textsc{Simon~J. Puglisi}, eds.,
  \emph{Proc.\ $46$th International Symposium on Mathematical Foundations of
  Computer Science (MFCS)}, volume 202 of \emph{LIPIcs}, pages 60:1--60:16.
  Schloss Dagstuhl, 2021.
\newblock \href{http://arxiv.org/abs/2107.02617}{\path{arXiv:2107.02617}}.

\bibitem[HY20]{HubacekY2020}
\textsc{Pavel Hub{\'{a}}cek} and \textsc{Eylon Yogev}.
\newblock \href{https://doi.org/10.1137/17M1118014} {\emph{Hardness of
  continuous local search: Query complexity and cryptographic lower bounds}}.
\newblock SIAM J. Comput., 49(6):1128--1172, 2020.
\newblock (Preliminary version in \emph{28th SODA}, 2017).
\newblock
  \href{https://eccc.weizmann.ac.il/eccc-reports/2016/TR16-063}{\path{eccc:2016/TR16-063}}.

\bibitem[Jer16]{Jerabek2016}
\textsc{Emil Jer{\'{a}}bek}.
\newblock \href{https://doi.org/10.1016/j.jcss.2015.08.001} {\emph{Integer
  factoring and modular square roots}}.
\newblock J. Comput.\ Syst.\ Sci., 82(2):380--394, 2016.
\newblock \href{http://arxiv.org/abs/1207.5220}{\path{arXiv:1207.5220}}.

\bibitem[JPY88]{JohnsonPY1988}
\textsc{David~S. Johnson}, \textsc{Christos~H. Papadimitriou}, and
  \textsc{Mihalis Yannakakis}.
\newblock \href{https://doi.org/10.1016/0022-0000(88)90046-3} {\emph{How easy
  is local search?}}
\newblock J. Comput.\ Syst.\ Sci., 37(1):79--100, 1988.
\newblock (Preliminary version in \emph{26th FOCS}, 1985).

\bibitem[Mor05]{Morioka2005}
\textsc{Tsuyoshi Morioka}.
\newblock
  \href{https://librarysearch.library.utoronto.ca/permalink/01UTORONTO\_INST/14bjeso/alma991106619311606196}
  {\emph{Logical approaches to the complexity of search problems: : proof
  complexity, quantified propositional calculus, and bounded arithmetic}}.
\newblock Ph.D. thesis, University of Toronto, Canada, 2005.

\bibitem[MP91]{MegiddoP1991}
\textsc{Nimrod Megiddo} and \textsc{Christos~H. Papadimitriou}.
\newblock \href{https://doi.org/10.1016/0304-3975(91)90200-L} {\emph{On total
  functions, existence theorems and computational complexity}}.
\newblock Theoret.\ Comput.\ Sci., 81(2):317--324, 1991.

\bibitem[Pap94]{Papadimitriou1994}
\textsc{Christos~H. Papadimitriou}.
\newblock \href{https://doi.org/10.1016/S0022-0000(05)80063-7} {\emph{On the
  complexity of the parity argument and other inefficient proofs of
  existence}}.
\newblock J. Comput.\ Syst.\ Sci., 48(3):498--532, 1994.

\bibitem[Sel06]{Selke2006}
\textsc{Joachim Selke}.
\newblock \href{https://citeseerx.ist.psu.edu/doc_view/pid/
  bcad5012db733066e2c5e37b9620a38732afe71d} {\emph{Autoreducibility and Friends
  About Measuring Redundancy in Sets}}.
\newblock Master's thesis, Gottfried-Wilhelm-Leibniz-Universität Hannover
  Fakultät für Elektrotechnik und Informatik Institut für Theoretische
  Informatik, 2006.

\bibitem[Tra70]{Trakhtenbrot1970}
\textsc{Boris~Avraamovich Trakhtenbrot}.
\newblock \href{http://mi.mathnet.ru/dan35484}
  {\emph{\foreignlanguage{russian}{Об автосводимости}
  ({R}ussian) [{O}n {A}utoreducibility]}}.
\newblock Dokl. Akad. Nauk SSSR, 192(6):1224--1227, 1970.
\newblock (English translation in \emph{Soviet Math. Dokl.} 11:814--817, 1970).

\end{thebibliography}

\ifitcs\else 
}\fi 

\end{document}